\documentclass{article}

\usepackage[utf8]{inputenc}
\usepackage{amsmath}
\usepackage{amsthm}
\usepackage{amssymb}
\usepackage[backgroundcolor=gray!20, linecolor=gray, textsize=scriptsize]{todonotes}
\usepackage{hyperref}
\usepackage[capitalize, nameinlink]{cleveref}
\usepackage{microtype}
\usepackage{graphicx}
\usepackage{subfigure}
\usepackage{booktabs} 

\usepackage[accepted]{icml2020}

\newtheorem{theorem}{Theorem}
\newtheorem{lemma}{Lemma}

\newtheorem{definition}{Definition}

\newtheorem{proposition}{Proposition}
\newtheorem{remark}{Remark}

\crefname{claim}{Claim}{Claims}
\crefname{fact}{Fact}{Facts}
\crefname{example}{Example}{Examples}

\newcommand{\E}{\textrm{\textbf{E}}}
\newcommand{\Var}{\textrm{\textbf{Var}}}
\newcommand{\e}{\textrm{{e}}}
\newcommand{\eps}{\varepsilon}
\newcommand{\norm}[1]{\left\lVert#1\right\rVert}

\renewcommand{\phi}{\varphi}

\newcommand{\poly}{\mathrm{poly}}
\renewcommand{\lg}{\log}
\newcommand{\km}{\texttt{$k$-means\,}}
\newcommand{\kmpp}{\texttt{$k$-means$++$\,}}
\newcommand{\kmpipe}{\texttt{$k$-means$||$\,}}
\renewcommand{\P}{\textrm{P}}

\emergencystretch=\maxdimen
\icmltitlerunning{Simple Analysis of k-meansII}

\begin{document}

\twocolumn[
\icmltitle{Simple and Sharp Analysis of \kmpipe}




\begin{icmlauthorlist}
\icmlauthor{Václav Rozhoň}{ETH}
\end{icmlauthorlist}

\icmlaffiliation{ETH}{ETH, Zurich}

\icmlcorrespondingauthor{Václav Rozhoň}{rozhonv@ethz.ch}

\icmlkeywords{$k$-means, clustering, distributed algorithms}

\vskip 0.3in
]



\printAffiliationsAndNotice{} 

\begin{abstract}

We present a simple analysis of \kmpipe (Bahmani et al., PVLDB 2012) -- a distributed variant of the k-means++ algorithm (Arthur and Vassilvitskii, SODA 2007). Moreover, the bound on the number of rounds is improved from  $O(\log n)$ to $O(\log n / \log\log n)$, which we show to be tight.

\end{abstract}

\section{Introduction}

Clustering is one of the classical machine learning problems. 
Arguably the simplest and most basic formalization of clustering is the \km formulation: we are given a (large) set of points $X$ in the Euclidean space and are asked to find a (small) set of $k$ centers $C$ so as to minimize the sum of the squared distances between each point and the closest center. That is, we minimize
\[
\phi_X( C) = \sum_{x \in X} \min_{c_i \in C} \Vert x - c_i \Vert^2,
\]
Due to its simplicity, $k$-means is considered as the problem that tests our understanding of clustering. 

The classical, yet still state-of-the-art algorithm \kmpp \cite{Arthur2007} combines two ideas to approach the problem. 
First, a fast randomized procedure finds a set of $k$ centers that by itself is known to be $O(\log k)$ competitive in expectation with respect to the optimal solution. 
Then, the classical Lloyd's algorithm \cite{Lloyd82} is run to improve the found solution until a local minimum is achieved. 

A significant disadvantage of the \kmpp is the inherent sequential nature of the first, seeding step: one needs to pass through the whole data $k$ times, each time to sample a single center. 
To overcome this problem, \cite{bahmani2012scalable} devised \kmpipe: a distributed version of the \kmpp algorithm. 

In their algorithm one passes through the dataset only few times to extract a set of roughly $O(k)$ candidate centers, from which one later chooses the final $k$ centers by the means of the classical \kmpp algorithm. 

\paragraph{Our contribution} In this work we first provide a new, simple analysis of \kmpipe, thus simplifying known proofs \cite{bahmani2012scalable} and \cite{bachem2017distributed}. 
In particular, if we denote by $\mu_X$ the mean of $X$ and $\phi^*$ the optimal solution, we prove in \cref{sec:simple} that $O(\log\frac{\phi_X(\mu_X)}{\phi^*})$ rounds of the \kmpipe algorithm suffice to get expected constant approximation guarantee. This matches known guarantees. 

In \cref{sec:sharp_upper}, we then proceed by refining the analysis to provide a better bound on the number of sampling rounds needed by the algorithm: we prove that $O(\log\frac{\phi_X(\mu_X)}{\phi^*} / \log\log \frac{\phi_X(\mu_X)}{\phi^*})$ rounds suffice. 
In \cref{sec:sharp_upper}, the second bound is shown to be tight for a wide range of values of $k$ by generalisation of a known lower bound \cite{bachem2017distributed}. 

The first analysis of \kmpipe \cite{bahmani2012scalable} invokes linear programming duality as a part of the argument. 
Its second analysis \cite{bachem2017distributed} is more similar to ours, as it only relies on basic lemmas known from the analysis of \kmpp \cite{arthur2007k}. Our one-page analysis is considerably shorter and, we believe, also simpler. 
It can be summed up as ``view one sampling step of the algorithm as a weighted balls into bins process''. We explain this in more detail in \cref{sec:simple}.

\section{Background and Notation}

We mostly adopt the notation of the paper \cite{bahmani2012scalable}. 
Let $X = \{x_1, \dots, x_n\}$ be a point set in the $d$-dimensional Euclidean space. We denote the standard Euclidean distance between two points $x_i, x_j$ by $\norm{x_i - x_j}$ and for a subset $Y \subseteq X$ we define the distance between $Y$ and $x \in X$ as $d(x,Y) = \min_{y \in Y} \norm{x-y}$. 

For a subset $Y \subseteq X$ we denote by $\mu(Y)$ its centroid, i.e., $\mu(Y) = \frac{1}{|Y|} \sum_{y \in Y} y$. 
For a set of points $C = \{c_1, \dots, c_k\}$ and $Y \subseteq X$ we define the \textit{cost} of $Y$ with respect to $C$ as $\phi_Y(C) = \sum_{y \in Y} d^2(y, C)$ and use a shorthand $\phi_x(C)$ for $\phi_{\{x\}}(C)$ and $\phi_X(c)$ for $\phi_X(\{c\})$. 
It is easy to check that for a given point set $A$, the center $x$ that minimizes its cost $\phi_A(x)$ is its mean $\mu_A$. 

The goal of the $k$-means problem is to find a set of \emph{centers} $C$, $|C|=k$ that minimizes the cost $\phi_X(C)$ for a given set of points $X$. 
The $k$-means problem is known to be NP-hard \cite{Aloise2009, mahajan2009planar} and even hard to approximate up to small constant precision \cite{awasthi2015hardness, lee2017improved}. 
From now on we fix an optimal solution $C^*$ and denote by $\phi^*$ its cost.
We assume that minimum distance between two points $x\not=y$ from $X$ is at least $1$ and that for the maximum distance $\Delta$ between two points from $X$ we have $\Delta = O(\log n)$; this implies $\phi_X(\mu_X) = \poly(n)$ and hence, the simplified complexities $O(\log n)$ and $O(\log n / \log \log n)$ in abstract. 


\subsection{\kmpp Algorithm} The classical \kmpp algorithm \cite{arthur2007k, ostrovsky2013effectiveness} computes the centers in $k$ sampling steps. 
After the first step where the first center is taken from uniform distribution, each subsequent step samples a new point from \emph{$D^2$-distribution}: if $C$ is the current set of centers, $x$ is being sampled with probability $\frac{\phi_x(C)}{\phi_X(C)}$, i.e., we sample the points proportional to their current cost. 

The \kmpp algorithm is known to provide an $O(\log k)$ approximation guarantee, in expectation. 
The analysis crucially makes use of the following two lemmas that we will also use. 

The first lemma tells us that if we sample uniformly a random point $p$ from some point set $A$, we expect the cost $\phi_A(p)$ to be within a constant factor of the cost $\phi_A(\mu_A)$ which is the smallest cost achievable with one center. 
One can think of $A$ as being a cluster of the optimal solution or the whole point set $X$. 

\begin{lemma}[Lemma 3.1 in \cite{Arthur2007}]
\label{lem:2apx}
Let $A$ be an arbitrary set of points. 
If we sample a random point $p \in A$ according to the uniform distribution, we have $\E[\phi_A(\{p\})] \le 2\phi_A(\mu_A)$. 
\end{lemma}

The second lemma ensures that up to a constant factor the same guarantee holds even for the $D^2$ distribution. 

\begin{lemma}[Lemma 3.2 in \cite{Arthur2007}]
\label{lem:8apx}
Let $A$ be an arbitrary set of points, $C$ be an arbitrary set of centers and $p \in A$ be a point chosen by $D^2$ weighting.
Then, $\E[\phi_A(C \cup \{p\})] \le 8\phi_A(\mu_A)$.
\end{lemma}

The analysis of \kmpp from the above two lemmas by a careful inductive argument: important observation is that sampling points proportional to their cost implies that we sample from an optimal cluster with probability proportional to its current cost, hence we preferably sample from costly clusters. Still, there is a small probability in each step that we sample from an already ``covered'' cluster and this is the reason why the final approximation factor is $O(\log k)$ rather than $O(1)$. 


\subsection{\kmpipe Algorithm} The distributed variant of the \kmpp algorithm called \kmpipe, was introduced in \cite{bahmani2012scalable}. 
The algorithm consists of two parts. In the first, \emph{overseeding} part (\cref{alg:kmpipe}), we proceed in $t$ sequential rounds after sampling uniformly a single center as in the first step of \kmpp. In each of the $t$ sampling rounds we sample each point of $X$ with probability $\min\left( 1, \frac{\ell \phi_x(C)}{\phi_X(C)} \right)$, i.e., $\ell \ge 1$ times bigger than the probability of taking the point in \kmpp, independently on the other points. 

\begin{algorithm}[tb]
    \caption{\kmpipe overseeding}
   \label{alg:kmpipe}
   {\bfseries Require}~data~$X$,\#~rounds~$t$,~sampling~factor~$\ell$
\begin{algorithmic}[1]
    \STATE Uniformly sample $x \in X$ and set $C = \{ x \}$.
    \FOR{$i \leftarrow 1, 2, \dots, t$}
        \STATE $C' \leftarrow \emptyset$
        \FOR{$x \in X$}
            \STATE Add $x$ to $C'$ with probability $\min\left( 1, \frac{\ell \phi_x(C)}{\phi_X(C)} \right)$
        \ENDFOR
        \STATE $C \leftarrow C \cup C'$
    \ENDFOR
    \STATE {\bfseries Return } $C$
\end{algorithmic}
\end{algorithm}

In the second part of the algorithm we collect the set of sampled centers $C$ and create a new, weighted, instance $X'$ of \km in which the weight of every center is equal to the number of points of $X$ to which the given center is the closest. 
The new instance is solved, e.g., with \kmpp as in \cref{alg:kmpipe2}. 
One can prove that finding a set $C$ with $\phi_X(C) = O(\phi^*)$ in \cref{alg:kmpipe} implies that the overall approximation guarantee is, up to a constant, the same as the approximation guarantee of the algorithm used in the second part of the algorithm (see \cite{bachem2017distributed}, proof of Theorem 1, or \cite{guha2003streaming}, in general), which is, in this case, $O(\log k)$ in expectation. 

\begin{algorithm}[tb]
    \caption{\kmpipe \cite{bahmani2012scalable}}
   \label{alg:kmpipe2}
    {\bfseries Require}~data~$X$,\#~rounds~$t$,~sampling~factor~$\ell$
    \begin{algorithmic}[1]
    \STATE $B \leftarrow $ Result of \cref{alg:kmpipe} applied to $(X,t,\ell)$
    \FOR{$c \in B$}
        \STATE $w_c \leftarrow$ \# of points $x \in X$ whose closest center in $B$ is $c$
    \ENDFOR
    \STATE $C \leftarrow$ Run \kmpp on the weighted instance $(B,w)$ 
    \STATE {\bfseries Return } $C$
\end{algorithmic}
\end{algorithm}

Hence, the analysis of \cref{alg:kmpipe2} boils down to bounding the number of steps $t$ needed by the \cref{alg:kmpipe} to achieve constant approximation guarantee for given sampling factor $\ell$. 
The authors of \cite{bahmani2012scalable} prove the following. 
\begin{theorem}[roughly Theorem 1 in \cite{bahmani2012scalable}]
\label{thm:main}
If we choose $t = O(\log\frac{\phi_X(\mu_X)}{\phi^*})$ and $\ell \ge k$, \cref{alg:kmpipe} gives a set $C$ with $\E[\phi_X(C)] = O(\phi^*)$. 
\end{theorem}
Their result was later reproved in \cite{bachem2017distributed}. We provide a new, simple proof in \cref{sec:simple}.

\subsection{Other Related Work}
\kmpp was introduced in \cite{Arthur2007} and a similar method was studied by \cite{ostrovsky2013effectiveness}. 
This direction led to approximation schemes \cite{jaiswal2012simple, jaiswal2015improved}, constant approximation results based on additional local search \cite{lattanzi2019better, we}, constant approximation bi-criteria results based on sampling more centers \cite{aggarwal2009adaptive, ailon2009streaming, wei2016constant},  approximate \kmpp based on Markov chains \cite{bachem2016approximate, bachem2016fast} or coresets \cite{bachem2018scalable}, analysis of hard instances \cite{Arthur2007, brunsch2013bad} or under adversarial noise \cite{bhattacharya2019noisy}. 
Consult \cite{celebi2013comparative} for an overview of different seeding methods for \km. 

There is a long line of work on a related $k$-median problem. 
A $k$-median algorithm of \cite{mettu2012optimal} is quite similar to \kmpipe and its analysis is also quite similar to our analysis. 
We discuss some other related work on \km more thoroughly at the end of \cref{sec:simple}. 

\section{Warm-up: a Simple Analysis}
\label{sec:simple}

In this section we provide a simple analysis of \cref{alg:kmpipe} based on viewing the process as a variant of the  \emph{balls into bins} problem. 
Recall that in the most basic version of the balls into bins problem, one throws $k$ balls into $k$ bins, each ball to a uniformly randomly chosen bin, and asks, e.g., what is a probability of a certain bin to be hit by a ball. This is equal to 
$
1 - \left( 1 - 1/k \right)^k \approx 1 - 1/\e, 
$
hence, we expect a constant proportion of the bins to be hit in a single step. 

To see the connection to our problem, we first define the notion of \emph{settled} clusters that is similar to notions used e.g. in \cite{aggarwal2009adaptive, lattanzi2019better}. 

\begin{definition}[Settled clusters]
We call a cluster $A$ of the optimum solution $C^*$ \textit{settled} with respect to current solution $C$, if $\phi_{A}(C) \le 10 \phi_{A}(C^*)$. 
Otherwise, we call $A$ \textit{unsettled} with respect to $C$. 

\end{definition}

We view the clusters of $C^*$ as bins and each sampling round of \cref{alg:kmpipe} as shooting at each bin and hitting it (i.e., making the cluster settled) with some probability. 
Intuitively, this probability is proportional to the cost of the cluster, since this is how we defined the probability of sampling any point of $X$. 
So, we view the process as a more general and repeated variant of the balls into bins process, where the costs of the clusters act like ``weights'' of the bins and we sample with probability (roughly) proportional to these weights. 
We prove now that clusters are being settled with probability roughly proportional to their cost (unless they are very costly). 

\begin{lemma}
\label{prop:make_settled}
Let $C$ be the current set of sampled centers and let $A$ be an unsettled cluster of the optimum solution. 
The cluster $A$ is \emph{not} made settled in the next iteration of \cref{alg:kmpipe} with probability at most
\[
\exp\left(-\frac{\ell \phi_{A}(C) }{ 5\phi_X(C)}\right). 
\]
\end{lemma}
Intuitively, for clusters $A$ with $\phi_A(C) \le \phi_X(C)/\ell$ the probability of hitting them in one step is of order $\frac{\ell \phi_A(C)}{\phi_X(C)}$ (using that $\e^{-x} \approx 1 -x$ for small positive $x$), while for more costly clusters the probability of hitting them is lower bounded by some constant. 
\begin{proof}
If we sample a point $c$ from $A$ according to $D^2$ weights, we have $\E[\phi_A(C\cup\{c\})] \le 8\phi_A^*$ by \cref{lem:8apx}. 
Hence, by Markov inequality, $A$ is made settled with probability at least $1 - \frac{8}{10} = \frac{1}{5}$. 
In other words, there is a subset of points $A' \subseteq A$ with $\phi_{A'}(C) \ge \phi_A(C)/5$ such that sampling a point from $A'$ makes $A$ settled. If $A'$ contains a point $x$ with $\phi_x(C) \ge \frac{\phi_X(C)}{\ell}$, we sample $x$ and make $A$ settled with probability $1$. Otherwise, we have
\begin{align*}
    \P(\text{$A$ does not get settled})
    &\le \prod_{x\in A'} \left( 1 - \ell \phi_x(C) / \phi_X(C) \right)\\
    &\le \exp(-\sum_{x \in A'} \ell \phi_x(C) / \phi_X(C) )\\
    &\le \exp(-\ell \phi_{A}(C) / (5 \phi_X(C)))
\end{align*}
where we used $1+x\le \e^{x}$ and $\phi_{A'}(C) \ge \phi_A(C)/5$. 
\end{proof}
Similarly to the classical balls into bins problem, we can now observe that the total cost of unsettled clusters drops by a constant factor in each step. 

From now on we simplify the notation and write $\phi^t_Y$ for the cost of the point set $Y$ after $t$ sampling rounds of \cref{alg:kmpipe}. Moreover, by $\phi_U^t$ we denote the total cost of yet unsettled clusters after $t$ sampling rounds. 

\begin{proposition}[roughly Theorem 2 in \cite{bahmani2012scalable}]
\label{prop:main}
Suppose that $\phi_X^{t} \ge 20\phi^*$. 
For $\ell \ge k$ we have
\[
\E\left[\phi_U^{t+1}\right] \le \left(1-\frac{1}{50}\right)\phi_U^t. 
\]
\end{proposition}
In other words,  the expected cost of yet unsettled clusters drops by a constant factor in each iteration. 
\begin{proof}
We split the unsettled clusters into two groups: a cluster $A$ with $\phi_{A}^t \ge \phi_U^t /(2k)$ we call \emph{heavy} and otherwise we call it \emph{light}. 
Note that the probability that a heavy cluster $A$ is not settled in $(t+1)$th iteration is by \cref{prop:make_settled} bounded by
\begin{align*}
\exp\left( \frac{-k\phi_{A}^t}{5\phi_X^t}\right)
\le  \exp\left( \frac{- k\phi_U^t}{10k\phi_X^t} \right)
\le  \exp\left( \frac{- 1}{20} \right)
\le \frac{24}{25}
\end{align*}
where we used that $\phi_U^t \ge \phi_X^t/2$: this holds since otherwise more than half the cost of $\phi_X^t$ is formed by settled clusters, hence $\phi_X^t < 20\phi^*$, contradicting our assumption $\phi_X^t \ge 20\phi^*$. 
Hence, after the sampling step, a heavy cluster $A$ does not contribute to the overall cost of unsettled clusters with probability at least $1/25$. 
This implies that the expected drop in the cost of unsettled clusters is at least
\begin{align*}
    \phi_U^t - \E[\phi_U^{t+1}] 
    \ge \sum_{\text{$A$ heavy}}   \phi_A^t / 25\\
    = \frac{1}{25} \left( \phi_U^t - \sum_{\text{$A$ light}}\phi_A^t \right)
    \ge \frac{\phi_U^t}{50}
\end{align*}
where we used that the light clusters have total cost of at most $k \cdot \phi_U^t / (2k) = \phi_U^t / 2$. 
\end{proof}

\cref{thm:main} now follows directly \cite{bahmani2012scalable, bachem2017distributed} and we prove it here for completeness. 


\begin{proof}[Proof of \cref{thm:main}]
From \cref{lem:2apx} it follows that after we sample a uniformly random point, we have $\E[\phi^0] \le 2 \phi_X(\mu_X)$. 
\cref{prop:main} then gives $\E[\phi_U^{t+1}] \le \frac{49}{50}\phi_U^t+ 20\phi^*$. 
Applying this result $T$ times, we get
\begin{align*}
\E[\phi^T_U] \le 2\left(\frac{49}{50}\right)^T\phi_X(\mu_X) + 20\phi^* \cdot \sum_{t=0}^{T-1} \left(\frac{49}{50}\right)^{t}\\
\le 2\left(\frac{49}{50}\right)^T\phi_X(\mu_X)
+ 1000 \phi^*
\end{align*}
Choosing $T = O(\log\frac{\phi_X(\mu_X)}{\phi^*})$ and recalling that $\phi^T \le \phi_U^T + 10\phi^*$ yields the desired claim. 
\end{proof}

\subsection{Additional Remarks}
For the sake of simplicity, we did not optimize constants and analysed \cref{alg:kmpipe} meaningfully only for the case $\ell =\Theta(k)$. 
In the following remarks we note how one can extend this (or some previous) analysis and then use it to compare \kmpipe more carefully to a recent line of work. 

\begin{remark}
With more care, the approximation factor in \cref{thm:main} can be made arbitrarily close to $8$. 
We omit the proof. 
\end{remark}

\begin{remark}
\label{rem:alpha}
With more care, for general $\ell = \alpha k$ one can prove that the number of steps of \cref{alg:kmpipe} needed to sample a set of points that induce a cost of $O(\phi')$ is $O(\log_{\alpha} \frac{\phi_X(\mu_X)}{\phi'})$ for $\alpha \ge 1$  and $O(\log \frac{\phi_X(\mu_X)}{\phi'} / \alpha)$ for $\alpha \le 1$. 
We omit the proof. 
\end{remark}

\subsection{Similar Algorithms with Additive Error}
\cref{rem:alpha} allows us to make a closer comparison of \kmpipe with a recent line of work of \cite{bachem2016fast, bachem2018scalable} that aims for very fast algorithms that allow for an additive error of $\eps \phi_X(\mu_X)$. 

According to \cref{rem:alpha}, to obtain such a guarantee for the oversampled set of centers, \cref{alg:kmpipe} needs to set $\ell = O(k)$ and sample for $t = \log(1/\eps)$ steps (this was observed in \cite{bachem2017distributed}) or it sets $\ell = O(k/\eps)$ points and sample just once (i.e., $t=1$). 
The approximation factor of \cref{alg:kmpipe2} is then multiplied by additional factor of $O(\log k)$ as this is the approximation guarantee of \kmpp. 

An approach of Bachem et al. \cite{bachem2018scalable} is similar to \kmpipe with $t=1$: there, authors propose a coreset algorithm that samples $\tilde{O}(dk/\eps^2)$ points from roughly the same distribution as the one used in the first sampling step of \cref{alg:kmpipe}. 
If we use their algorithm by running \kmpp on the provided coreset, we get an algorithm with essentially the same guarantees as \cref{alg:kmpipe2} with number of rounds $t=1$ and $\ell = O(k/\eps)$. 
The main difference is that in \cref{alg:kmpipe2}, the weight of each sampled center used by \kmpp subroutine is computed as the number of points for which the center is the closest, whereas in the coreset algorithm, each center is simply given a weight inversely proportional to the probability that the center is sampled. 
This allows the coreset algorithm to be faster than \cref{alg:kmpipe2} with $t=1$ and $\ell = O(k/\eps)$, whose time complexity in this case is $O(nk/\eps)$. This is at the expense of higher number of sampled points. 

A paper of \cite{bachem2016fast} uses the Metropolis algorithm on top of the classical \kmpp algorithm to again achieve additive $\eps \phi_X(\mu_X)$ (and multiplicative $O(\log k)$) error, while sampling only $O(\frac{k}{\eps}\log\frac{k}{\eps})$ points from the same distribution as \cref{alg:kmpipe2} with $t=1$ and $\ell=O(k/\eps)$.  
While the number of taken samples is only slightly higher than the one of \cref{alg:kmpipe2} with $t=1, \ell=O(k/\eps)$, their running time is much better $\tilde{O}(k^2 / \eps)$. 

We see that the main advantage of \kmpipe lies in the possibility of running multiple, easily distributed, sampling steps that allow us to achieve strong guarantees. 

\subsection{MPC Context and a Theoretical Implication}
Massively Parallel Computing (MPC) \cite{dean2004mapreduce, karloff2010model} is a distributed model in which the set of input points $X$ is split across several machines, each capable of storing $\tilde{O}(s)$ points for some parameter $s$. 
In each round, each machine performs some computation on its part of input and sends $\tilde{O}(s)$ bits to other machines such that each machine receives $\tilde{O}(s)$ bits in total. We want to minimize the number of rounds while keeping the memory $s$ small. 

\cref{alg:kmpipe2} with $t$ sampling steps can be implemented in $O(t)$ MPC rounds if $s = \Omega(tk)$. 
Hence, we get an $O(1)$-approximation algorithm for $k$-means in $O(\log \frac{\phi_X(\mu(X))}{\phi^*}) = O(\log n)$ rounds if each machine has $\tilde{\Theta}(k)$ memory. 

In the case of per machine memory $O(n^\alpha k)$ for constant $0 < \alpha \le 1$, we can use \kmpipe to achieve constant approximation guarantee in only $O(1/\alpha)$ rounds, as was noted by \cite{mohsen}. 
A hierarchical clustering scheme of Guha et al. \cite{guha2003streaming} in this setting needs $O(1/\alpha)$ rounds to output a set of $k$ centers $\bar{C}$, but at the expense of approximation guarantee: we have $$\phi_X(\bar{C}) = 2^{O(1/\alpha)} \phi_X(C^*).$$ 
\kmpipe can now recover this loss in approximation: if we run \cref{alg:kmpipe} but start not with the trivial $C_0 = \{x\}$ for $x$ sampled uniformly at random, but instead take $C_0 = \bar{C}$, \cref{prop:main} ensures that $$t = O(\log\frac{\phi_X(\bar{C})}{\phi_X(C^*)}) = O(\log 2^{O(1/\alpha)}) = O(1/\alpha)$$ sampling rounds suffice to get a clustering with $O(1)$-approximation guarantee.

\subsection{Submodular Context}
One can observe why $O(\log\frac{\phi_X(\mu_X)}{\phi^*})$ is a natural bound for the number of rounds by considering a different way of achieving the same round complexity. 
First, note that after sampling the first point $c_1$, we have $\E[\phi_X(c_1)] \le 2\phi_X(\mu_X) = 2\phi_X(\mu_X)$ via \cref{lem:2apx} with the set $A$ chosen as the whole set $X$. 
The process of adding new points to the solution now satisfies a natural law of diminishing returns: for any $C_1 \subseteq C_2 \subseteq X$ and $c \in X$ we have
\begin{align*}
&\phi_X(\{c_1\}\cup C_1 \cup \{c\}) - \phi_X(\{c_1\}\cup C_1)\\
&\ge \phi_X(\{c_1\}\cup C_2 \cup \{c\}) - \phi_X(\{c_1\}\cup C_2)
\end{align*}

In other words, the function $\phi_X(\{c_1\}) - \phi_X(\{c_1\} \cup C)$ is submodular (see e.g. \cite{krause2014submodular} for the collection of uses of submodularity in machine learning). 
Then one can use recent results about distributed algorithms for maximizing submodular functions (see e.g. \cite{sub6, sub4, sub3, sub2}) to get that in $O(1)$ distributed rounds, one can find a set of $k$ points $C$ such that 
$
\phi_X(\{c_1\} \cup C) - \phi^* \le \left( \phi_X(\{c_1\}) - \phi^* \right) /2,
$
i.e., the distance to the best solution drops by a constant factor.  
Continuing the same process for $\log((\phi_X(\mu_X)) / \phi^*)$ rounds, one gets the same theoretical guarantees as with running \cref{alg:kmpipe2}. 
However, the advantage of \kmpipe is its extreme simplicity and speed. 
Moreover, rather surprisingly, we prove in the next section that the asymptotical round complexity of \kmpipe is actually slightly better than logarithmic. 

\section{Sharp Analysis of \kmpipe: Upper Bound}
\label{sec:sharp_upper}

It may seem surprising that the analysis of \cref{thm:main} can be strengthened, since even for the classical balls into bins problem, where we hit each bin with constant probability, we need $O(\log k)$ rounds to hit all the bins with high probability. 
However, during our process we can disregard already settled clusters since they are not contributing substantially to the overall cost. 
If we go back to the classical balls into bins problems and let that process repeat on the same set of bins, with the additional property that in each round we throw each one of $k$ balls to a random bin \emph{out of those that are still empty}, we expect to hit all the bins in mere $O(\log^* n)$ steps \cite{lenzen2011tight}. 
Roughly speaking, this is because  of the rapid decrease in the number of bins: after the first round, the probability of a bin remain empty is roughly $\frac{1}{2}$, but after second round it is only roughly $\frac{1}{2^2}$ since the number of bins decreased, in the next iteration it is even roughly $\frac{1}{2^{2^2}}$ and so on. 
In our weighted case we cannot hope for such a rapid decrease in the number of bins, since the costs of clusters can form a geometric series, in which case we get rid of only a small number of clusters in each step (cf. \cref{sec:sharp_lower}) 
\footnote{We believe it is an interesting problem to analyze whether there are reasonable assumptions on the data under which the round complexity indeed follows $\log^* n$ behaviour (for example, this is the case if clusters are of same size and lie in vertices of a simplex). This could explain why in practice it is enough to run \cref{alg:kmpipe} only for few rounds \cite{bahmani2012scalable}.  }.

In this section we show a more careful analysis that bounds the number of necessary steps to 
\[
O\left( \log \frac{\phi_X(\mu_X) }{ \phi^*  } /  \log\log \frac{\phi_X(\mu_X)  }{ \phi^* } \right). 
\]
In the rest of the paper we use the notation $ \gamma = \phi_X(\mu_X) / \phi^*$. 

\paragraph{Concentration} The number of steps $O(\log\gamma / \log\log \gamma)$ suffice even with high probability; to prove it, we first recall the classical Chernoff bounds that are used to argue about concentration around mean. 
\begin{theorem}[Chernoff bounds]
\label{thm:chernoff}
Suppose $X_1, \dots, X_n$ are independent random variables taking values in $\{0, 1\}$. Let $X$ denote their sum. Then for any $0 \le \delta \le 1$ we have
\[
\P(X \le (1-\delta)\E[X]) \le \e^{-\E[X]\delta^2 / 2},
\]
\[
\P(X \ge (1+\delta)\E[X]) \le \e^{-\E[X]\delta^2 / 3},
\]
and for $\delta \ge 1$ we have
\[
\P(X \ge (1+\delta)\E[X]) \le \e^{-\E[X]\delta / 3}.
\]
\end{theorem}

\paragraph{Intuition} In the following \cref{prop:main_fine}, a refined version of \cref{prop:main}, we argue similarly, but more carefully, about one sampling step of \cref{alg:kmpipe}. 
The difference is that we analyze not only the drop in the cost of unsettled clusters, but also the drop in the \emph{number} of unsettled clusters. 
Here is the intuition. 

Let us go back to the proof of  \cref{prop:main}, where we distinguished heavy and light clusters. Heavy clusters formed at least constant proportion of the cost of all clusters and every heavy cluster was hit with probability that we lower-bounded by $1/25$. 
For a light cluster we cannot give a good bound for the probability of hitting it, but since it is not very probable that one light cluster is hit by more than one sampled center, if we denote $\alpha = \left(\sum_{\text{$A$ light}} \phi_A^t\right) / \phi_U^t$, i.e., $\alpha$ is the proportional cost of the light clusters, we expect roughly $\alpha k$ clusters to become settled. 
On the other hand, we may, optimistically, hope that after one iteration the cost of unsettled clusters drops from $\phi^t$ to $\alpha \phi^t$, since the heavy clusters are hit with high probability. 
This is not exactly the case, since for a heavy cluster we have only a constant probability of hitting it. 
But we can consider two cases: either there are lot of heavy clusters and we again make $\alpha k$ clusters settled, or their cost is dominated by few, \emph{massive} clusters, each of which is not settled only with exponentially small probability and, hence, we expect the cost to drop by $\alpha$ factor. 

The tradeoff between the behaviour of light and heavy clusters then yields a threshold $\alpha \approx 1/\poly\log\gamma$ that balances the drop of the cost and of the number of unsettled clusters. 


\begin{proposition}
\label{prop:main_fine}
Suppose that after $t$ steps of \cref{alg:kmpipe} for $\ell=k$ there are $k_t$ unsettled clusters and their total cost is $\phi^t_U$. 
Assume that $\phi_X^t \ge 20\phi^*$. 
After the next sampling step, with probability at least $1 - \exp(\Theta(k^{0.1}))$, we have that either the number of unsettled clusters decreased by at least $k/(40\sqrt{\log\gamma})$, or the total cost of unsettled clusters decreased from $\phi_U^t$ to at most $4\phi_U^t/\sqrt[3]{\log\gamma}$. 
\end{proposition}
\begin{proof}
Note that $\phi_X^t \ge 20\phi^*$ implies $\phi_U^t \ge \phi_X^t / 2$, since otherwise, more than half the cost of $\phi_X^t$ would be formed by settled clusters and, hence, $\phi_X^t < 20\phi^*$, a contradiction. 

We will say that an unsettled cluster is \emph{heavy} if its cost is at least $\phi_A^t \ge \phi_U^t / k$ and \emph{light} otherwise. 
Let $\alpha = (\sum_{\text{$A$ light}} \phi^t_A) / \phi^t_U$, i.e., the proportional cost of the light clusters. 
We will distinguish three possible cases. 

\begin{enumerate}
    \item $\alpha \ge 1 / \sqrt{\log\gamma }$,
    \item $\alpha < 1 / \sqrt{\log\gamma }$ and there are at least $k / \sqrt{\log\gamma }$ heavy clusters,
    \item $\alpha < 1 / \sqrt{\log\gamma }$ and there are less than $k / \sqrt{\log\gamma }$ heavy clusters. 
\end{enumerate}
For each case we now prove that with probability $1 - \exp(\Omega(-k^{0.1}))$ we either settle at least $k/(40\sqrt{\log\gamma})$ clusters or the total cost of unsettled clusters drops from $\phi_U^t$ to $4\phi_U^t/\sqrt[3]{\log\gamma}$. 

\begin{enumerate}
    \item 
    By \cref{prop:make_settled}, each light cluster $A$ gets settled with probability at least 
    \[
    1 - \e^{-k \phi^t_{A} / (5\phi_X^t)} \ge (k \phi^t_{A}) / (20\phi_U^t),
    \]
    using $\phi^t_U \ge \phi_X^t/2$ and $\e^{-x} \le 1 + x/2$ for $0 \le x \le 1$. 
If we define $X_A$ to be an indicator of whether a light cluster $A$ got settled in this iteration and $X = \sum_{\text{$A$ light}} X_A$, we have
\begin{align*}
\E[X] = 
\sum_{\text{$A$ light}} \E[X_A]
\ge \sum_{\text{$A$ light}} \frac{k \phi_A^t}{20 \phi_U^t}\\
= \alpha k  / 20
\ge \frac{k}{20\sqrt{\log\gamma}}
\end{align*}
where we used our assumption on $\alpha$. 
Invoking the first bound of \cref{thm:chernoff}, we get 
\[
\P(X \le \E[X]/2) \le \e^{-\E[X] / 8} \le \e^{-\Theta(k^{0.1})},
\]
using that $k \ge \lg\gamma/\lg\lg\gamma$. 


\item 
We proceed analogously to the previous case. 
By \cref{prop:make_settled} we get that each heavy cluster gets settled with probability at least $$1 - \e^{-k\phi^t_A / (5\phi_X^t)} \ge 1 - \e^{-1/10} \ge \frac{1}{20},$$ using $\phi^t_U \ge \phi_X^t/2$ and the definition of heavy cluster. 

We define $X_A$ to be an indicator of whether a heavy cluster $A$ got settled in this iteration and $X = \sum_{\text{$A$ heavy}} X_A$. We have
\[
\E[X] = \sum_{\text{$A$ heavy}} \E[X_A]
\ge k / \sqrt{\log\gamma } \cdot \frac{1}{20}
= \frac{k}{20 \sqrt{\log\gamma } }.
\]
Invoking the first bound of \cref{thm:chernoff}, we get 
\[
\P(X \le \E[X]/2) \le \e^{-\E[X] / 8} \le \e^{-\Theta(k^{0.1})},
\]
using that $k \ge \lg\gamma/\lg\lg\gamma$. 

\item 
Let $\zeta = \sqrt[10]{\lg \gamma}\frac{\phi_U^t}{k}$. 
We call a heavy cluster \emph{massive} if its cost is at least $\zeta$. Since we know that there are at most $\frac{k}{\sqrt{\log\gamma }}$ heavy clusters, the total cost of clusters that are heavy but not massive is at most
\[
\frac{k}{\sqrt{\log\gamma }} \cdot  \sqrt[10]{\lg\gamma }\frac{\phi^t_U}{k}
\le \frac{\phi^t_U}{\sqrt[3]{\log\gamma }}
\]
Hence, the total contribution of clusters that are not massive is at most 
\[
\frac{\phi_U^t}{\sqrt[3]{\log\gamma }} + \alpha\phi_U^t 
\le \phi_U^t ( \frac{1}{\sqrt[3]{\log\gamma }} + \frac{1}{\sqrt{\log\gamma }}) 
\le \frac{2\phi_U^t}{\sqrt[3]{\log\gamma }}. 
\]


By \cref{prop:make_settled} each massive cluster is not settled with probability at most $ \e^{-k\phi_A^t/(5\phi_X^t)} $. 
Define the random variable $X_A$ to be equal to $0$ if a massive cluster $A$ gets settled in this iteration and $\phi_A^t / \zeta$ otherwise. 
Let $X = \sum_{\text{$A$ massive}} X_A$. 
Note that expected cost of massive clusters that are not settled in this iteration is bounded by $\E[X] \cdot \zeta$. 

The value of $X$ is stochastically dominated by the value of a variable $X'$ defined as follows. 
We first replace each $X_A$ by some variables $X'_{A, 1}, \dots, X'_{A, \lceil \phi^t_A/\zeta\rceil}$, each new variable $X'_{A, i}$ being equal to $1$ with probability $\e^{-k \zeta / (10\phi_X^t)}$ and zero otherwise, independently on the other variables $X'_{A,j}$. 
Note that the sum $X'_A = X'_{A, 1} + \dots + X'_{A, \lceil \phi^t_A/\zeta\rceil}$ stochastically dominates the value of $X_A$, since it attends the value $\lceil \frac{\phi_A^t}{\zeta} \rceil \ge \frac{\phi_A^t}{\zeta}$ with probability 
\begin{align*}
\prod_{i=1}^{\lceil \phi^t_A / \zeta \rceil} \e^{-k\zeta / (10\phi_X^t)} \ge \left( \e^{-k\zeta / (10\phi_X^t)} \right)^{ 2\phi^t_A / \zeta} \\
= \e^{-k\phi_A^t/(5\phi_X^t)},
\end{align*}
and otherwise is nonnegative. 
Hence, the value $X' = \sum X'_A$ stochastically dominates the value $X$. 

Since the number of variables $X'_{A,i}$ is $\sum_A \lceil \frac{\phi_A^t}{\zeta} \rceil \le \frac{2\phi_X^t}{\zeta}$ we have 
\begin{align*}
\E[X] &\le \E[X'] 
\le \frac{2\phi_X^t}{\zeta} \cdot \exp(-k\zeta /(10\phi_X^t))\\
&\le \frac{4\phi_U^t}{\zeta} \cdot \exp(\frac{-k \sqrt[10]{\log\gamma}\phi^t_U}{20k\phi_U^t} )\\
&=\frac{4\phi_U^t}{\sqrt[10]{\log\gamma}\phi^t_U / k} \cdot \exp(\frac{-\sqrt[10]{\log\gamma}}{20} )
\le \frac{k}{\sqrt{\log\gamma}}
\end{align*}
where the last step assumes $\gamma$ large enough. 
Finally, we bound
\begin{align*}
&\P(X \ge 2k / \sqrt{\log \gamma}) 
\le \P(X' \ge 2k / \sqrt{\log \gamma})\\
&= \P(X' \ge (1+\delta)\E[X'])
\end{align*}
for $\delta = \frac{2k }{ \sqrt{\log \gamma} \E[X']} - 1 \ge \frac{k }{ \sqrt{\log \gamma} \E[X']} \ge 1$ by above bound on $\E[X']$. 
Applying the third bound in \cref{thm:chernoff}, we conclude that
\begin{align*}
&\P(X \ge 2k / \sqrt{\log \gamma}) 
\le \exp(-\E[X']\delta/3)\\
&\le \exp(- k/(3\sqrt{\log \gamma}))
\le \e^{-\Omega(k^{0.1})}. 
\end{align*}
Hence, with probability $1 - \e^{-\Omega(k^{0.1})}$ the total cost of clusters that remain unsettled after this iteration is bounded by
\begin{align*}
     \frac{2\phi_U^t}{\sqrt[3]{\log\gamma }} + \frac{2k}{\sqrt{\log\gamma}} \cdot \zeta 
     \le \frac{4\phi_U^t}{\sqrt[3]{\log\gamma }}
\end{align*}
\end{enumerate}
\end{proof}

\begin{theorem}
\label{thm:main_fine}
For $\ell=k$, \cref{alg:kmpipe} achieves a constant approximation ratio for the number of sampling steps $t = O(\min(k, \frac{\lg\gamma}{\lg\lg\gamma}))$ steps with probability $1 - \exp(\Omega(k^{0.1}))$.  
\end{theorem}
\begin{proof}
To see that the number of steps is bounded by $k$ with high probability, note that by \cref{prop:make_settled} each unsettled cluster is made settled with probability at least $1 - \e^{-k\phi_A^t/(5\phi_X^t)}$. So, unless $\phi_X^t \le 20\phi^*$, we have $\phi_U^t \ge \phi_X^t/2$ and, hence, with probability at least 
\[
1 - \prod_{\text{$A$ unsettled}} \exp(-k\phi_A^t/10\phi_U^t) =
1 - \exp(-k/10)
\]
we make at least one cluster settled. Union bounding over first $k$ steps of the algorithm, we conclude that the algorithm finishes in at most $k$ steps with probability $1 - \exp(-\Omega(k))$. 

Whatever point $c_0$ is taken uniformly at the beginning of \cref{alg:kmpipe}, for the next iteration we invoke \cref{prop:make_settled} with $A=X$ (in its formulation we say that $A=X$ is unsettled cluster) to conclude that with probability $1 -  \exp(-\Omega(k))$ we have $\phi_X^1 \le 10\phi_X(\mu_X)$. 



We invoke \cref{prop:main_fine} and union bound over $k$ subsequent iterations of the algorithm to conclude that with probability $1 - \exp(\Omega(k^{0.1}))$, in each sampling step of \cref{alg:kmpipe} we either make at least $k/(40\sqrt{\log\gamma })$ clusters settled or the cost of unsettled clusters decreases by a factor of $4/\sqrt[3]{\log\gamma }$. 
The first case can happen at most $O(\sqrt{\lg\gamma}) = O(\lg\gamma / \lg\lg\gamma)$ times , whereas the second case can happen at most $O(\lg_{\sqrt[3]{\log\gamma }/4} \gamma) = O(\lg\gamma / \lg\lg\gamma)$ times, until we have $\phi^t_X \le 20\phi^*$. Hence, the algorithm achieves constant approximation ratio after $O(\lg\gamma/\lg\lg\gamma)$ steps, with probability $1 - \exp(\Omega(k^{0.1}))$. 
\end{proof}
\begin{remark}
The high probability guarantee in \cref{thm:main_fine} can be made $1 - \exp(\Omega(k^{1-o(1)}))$. We omit the proof. 
\end{remark}
\section{Sharp Analysis of \kmpipe: Lower Bound}
\label{sec:sharp_lower}
In this section we show that the upper bound of $O(\log\gamma / \log\log\gamma)$ steps is the best possible. 
Note that \cite{bachem2017distributed} proved that for $\ell=k$ there is a dataset $X$ such that $\E[\phi^t_X] \ge \frac14 (4kt)^{-t} \Var(X)$ for $t < k-1$. Hence, for $k = O(\poly(\log\gamma))$ we conclude that for the number of steps $T$ necessary to achieve constant approximation we have $(T \poly(\log\gamma))^T \ge \gamma$, implying $T = \Omega(\lg\gamma / \lg\lg\gamma)$. We complement their result by showing that the same lower bound also holds for $k = \Omega(\poly(\log\gamma))$. 

In the following theorem we construct an instance where the distance of two different data points can be zero, but it is easy to generalize the result for the case where we have the distance between two different points lower bounded by $1$. 

\begin{theorem}
For any function $f$ with $f(x) = \Omega(\log^{10} x), f(x) = O(x^{0.9})$, there is a dataset $X$ with $|X| = \Theta(k)$ and $k = \Theta(f( \phi_X(\mu_X)))$ such that $\phi^* = 0$ and with probability arbitrarily close to $1$ \cref{alg:kmpipe} needs $\Omega(\lg(\phi_X(\mu_X)) / \lg\lg(\phi_X(\mu_X)))$ iterations to achieve cost zero. 
\end{theorem}
\begin{proof}
First we describe the dataset $X$.
We place $|X|-k+1$ points $x_{0,j}$ for $1 \le j \le |X|-k+1$ to the origin, i.e., $x_{0,j} = \textbf{0}$. We choose $|X|=\Theta(k)$ to be of such size that we know that with probability at least $1-\eps$, for a given constant $\eps$, the first uniformly chosen center is $\textbf{0}$. 

For each one of the remaining $k-1$ points we consider a new axis orthogonal to the remaining $k-2$ axes and place the point on this axis. For $x_{i,j}$, $1 \le i \le T$, $1 \le j \le k/T$, we set $\Vert x_{i,j} - \textbf{0} \Vert^2 =  L^{2(T-i+1)}$, for some large enough $L$ and $T = \Theta(L / \log L)$ with the multiplicative constant chosen in such a way that for $\phi_X(\mu_X) = \Theta(\frac{k}{T} \cdot L^{2T})$ we have $\phi_X(\mu_X) = \Theta(k \cdot 2^L)$. Note that we need to choose $L = \Omega(\poly\log k)$ to achieve $k = \Theta(f(\phi_X(\mu_X))) = O((\phi_X(\mu_X))^{0.9})$. 
We define $\phi_i^t = \sum_{j=1}^{k/T} \phi^{t}_{x_{i,j}}$. 

Conditioning on the first uniformly taken point being $x_{0,j}$ for some $j$, we prove by induction for $0 \le t \le T$ that with probability at least $1 - t \cdot 1/\poly(k)$, after $t$ sampling steps of the algorithm we have for each $i > t$ that out of $k/T$ points $x_{i,j}$, for some $j$, at most 
\[
\frac{k}{T} \cdot \left(\frac23\right)^{i-t}
\]
of them have been sampled as centers. 
This will prove the theorem, since it implies that with probability at least $1 - 1/\poly(k)$, after $t=T=\Theta(L/\log L)$ steps the cost $\phi^{t}$ is still nonzero; since we have $k = O((\phi_X(\mu_X))^{0.9})$, we have
\begin{align*}
\Theta(L/\log L) = \Theta(\lg \frac{\phi_X(\mu_X)}{k} / \lg\lg\frac{\phi_X(\mu_X)}{k})\\
= \Theta(\lg \phi_X(\mu_X) / \lg\lg\phi_X(\mu_X)).
\end{align*}

For $t = 0$ the claim we are proving is clearly true. 
For $t \ge 1$ note that by induction with probability at least $1 - (t-1)/\poly(k)$ we have that at least 
\[
k/T - k/T \left(\frac23\right)^{} = k/(3T)
\]
of the points $x_{t,j}$, for some $j$, were not sampled as centers. 
Hence, $\phi_X^t \ge \frac{k}{3T} L^{2(T-t+1)}$ and the probability that each point $x_{i,j}$, $i > t$, is being sampled is bounded by 
\[
\frac{k \cdot L^{2(T-i+1)}}{k L^{2(T-t+1)}/(3T)} 
= \frac{3T}{L^{2(i-t)}}
\le \frac{1}{L^{i-t}}
\]
for $L$ large enough. 
Hence, for any $i > t$, the expected number of points $x_{i,j}$ that are being hit is bounded by $k / (T\cdot L^{i-t})$. To get concentration around this value for given $i$, consider two cases. 
\begin{enumerate}
    \item $\frac{k}{3^{i-t}T} \ge k^{0.1}$. Then, by the third bound of \cref{thm:chernoff} we can bound the probability of taking more than $\frac{k}{3^{i-t}T}$ clusters by $\exp(-\Theta(\frac{k}{3^{i-t}T})) \le \exp(-k^{0.1}) \le 1 / \poly(k)$. 
    \item $\frac{k}{3^{i-t}T} < k^{0.1}$. Using the assumption $k=\Omega(\log^{10}\phi_X(\mu_X))$, hence $k=\Omega(L^{10})$, we have $3^{i-t} > k^{0.9}/T \ge k^{0.7}$. For $L$ sufficiently large we then have $L^{i-t} \ge \poly(k)$ for any fixed polynomial, hence, the expected number of hits is at most $k / (T\cdot L^{i-t}) \le 1/\poly(k)$, hence, with probability at least $1 - 1/\poly(k)$ there is no hit. 
\end{enumerate}
In both cases, conditioning on an event of probability $1 - 1/\poly(k)$, for any $i > t$ the number of points $x_{i,j}$ that were sampled as centers is with probability at least $1 - t/\poly(k)$ bounded by 
\[
\frac{k}{T} \left( \left(\frac23\right)^{i-t + 1} + \left(\frac13\right)^{i-t} \right)
\le \frac{k}{T} \cdot \left(\frac23\right)^{i-t} 
\]
as needed. 
\end{proof}

\section{Acknowledgement}

I would like to thank my advisor Mohsen Ghaffari for his very useful insights and suggestions regarding the paper. 
I also thank Davin Choo, Danil Koževnikov, Christoph Grunau, Andreas Krause, and Julian Portmann for useful discussions and anonymous referees for helpful comments.  

This project has received funding from the European Research Council (ERC) under the European Union’s Horizon 2020 research and innovation programme (grant agreement No. 853109)

\bibliography{ref}
\bibliographystyle{icml2020}

\end{document}